\documentclass{article}
\usepackage[margin=1.25in]{geometry}
\usepackage{graphicx}
\usepackage{amsmath}
\usepackage{amssymb}
\usepackage{mathtools}
\usepackage{parskip}
\usepackage[hidelinks]{hyperref}
\usepackage{graphicx}
\usepackage{natbib}
\usepackage[bf]{caption}
\usepackage{subcaption}

\usepackage{color}
\usepackage{amssymb}
\usepackage{amsthm}
\newtheorem{proposition}{Proposition}
\newtheorem{corollary}{Corollary}
\begin{document}
\title{
	On modelling bicycle power-meter measurements:~Part~II \\
	{\Large Relations between rates of change of model quantities}
}
\author{
	Tomasz Danek%
	\footnote{
		AGH--University of Science and Technology, Krak\'{o}w, Poland, \texttt{tdanek@agh.edu.pl}
	}\,,
	Michael A. Slawinski%
	\footnote{
		Memorial University of Newfoundland, Canada, \texttt{mslawins@mac.com}
	}\,,
	Theodore Stanoev%
	\footnote{
		Memorial University of Newfoundland, Canada, \texttt{theodore.stanoev@gmail.com}
	}
}
\date{December~23, 2020}
\maketitle
\begin{abstract}
Power-meter measurements are used to study a model that accounts for the use of power by a cyclist.
The focus is on relations between rates of change of model quantities, such as power and speed, both in the context of partial derivatives, where other quantities are constant, and Lagrange multipliers, where other quantities vary to maintain the imposed constraints.
\end{abstract}
\section{Introduction}
Using power meters to study cycling performance allows us to gain quantitative information about relations whose qualitative aspects are known based on observations; for instance, riding with a given speed with a tailwind requires less effort than riding with the same speed against a headwind.
The quantification of such a relation, however, is necessary to proceed with an optimization to achieve\,---\,under constraints imposed by the capacity of a cyclist\,---\,the least time to cover a distance that is subject to winds and contains flats, hills and descents.

Many studies examine the physics of cycling.
This article is the second part of~\citet{DSSbici1}, which also contains the pertinent bibliography.
Herein\,---\,using a model relating power-meter measurements to the motion of a bicycle, examined by~\citet{DSSbici1}\,---\,we formulate expressions that allow us to quantify relations between the rates of change of parameters contained within this model.

We begin this paper by presenting the expression to account, by modelling, for the values measured by a power meter.
Using this expression and the implicit function theorem, we derive explicit expressions of the rates between the model parameters.
We complete this paper by interpreting and comparing quantitive results based on power-meter and GPS measurements collected on a flat course and an inclined course.
Both are in Northwestern Italy; the former is between Rivalta Bormida and Pontechino, in Piemonte; the latter is between Rossiglione and Tiglieto, in Liguria.
In the appendices, we compare the flat-course results to optimizations based on the Lagrange multipliers, and comment on the bijection between the generated power and the bicycle speed.
\section{Formulation}
\subsection{Measurements and model}
\label{sec:MM}
To account for the cyclist's use of the power measured by a power meter, we consider,
\begin{equation}
	\label{eq:formula}
	P
	=
	F_{\!\leftarrow}\,V_{\!\rightarrow}
	\,,
\end{equation}
where $P$ stands for the value of the required power, $F_{\!\leftarrow}$ for the forces opposing the motion, and $V_{\!\rightarrow}$ for the ground speed of the bicycle.

Explicitly, we assume that \citep[e.g.,][]{DSSbici1}
\begin{equation}
	\label{eq:P}
	P
	=
	\underbrace{
		\frac{
			mg\sin\theta
			+
			m\,a
			+
			{\rm C_{rr}}mg\cos\theta
			+
			\tfrac{1}{2}\,\eta\,{\rm C_{d}A}\,\rho\,
			\left(
				V_{\!\rightarrow} + w_{\leftarrow}
			\right)^{2}
		}{
			1-\lambda
		}
	}_{F_{\!\leftarrow}}
	V_{\!\rightarrow}\,;
\end{equation}
herein, $m$ is the mass of the cyclist and the bicycle, $g$ is the acceleration due to gravity, $\theta$ is the slope, $a$ is the change of speed, $\rm C_{rr}$ is the rolling-resistance coefficient, $\rm C_{d}A$ is the air-resistance coefficient, $\rho$ is air density, $V_{\!\rightarrow}$ is the ground-speed of the bicycle, $w_{\leftarrow}$ is the wind component opposing the motion, $\lambda$ is the drivetrain-resistance coefficient, $\eta$ is a quantity that ensures the proper sign for the tailwind effect, $w_{\leftarrow}<-V_{\!\rightarrow}\iff\eta=-1$\,, otherwise, $\eta=1$\,; throughout this work,~$\eta=1$\,.

To estimate quantities that appear on the right-hand side of equation~(\ref{eq:P})\,---\,specifically, $\rm C_{d}A$\,, $\rm C_{rr}$ and $\lambda$\,---\,given the measurement,~$P$\,, we write
\begin{equation}
	\label{eq:model}
	f
	=
	P
	-
	\underbrace{
		\frac{
			mg\sin\theta\,V_{\!\rightarrow}
			+
			m\,a\,V_{\!\rightarrow}
			+
			{\rm C_{rr}}mg\cos\theta\,V_{\!\rightarrow}
			+
			\tfrac{1}{2}\,\eta\,{\rm C_{d}A}\,\rho
			\left(V_{\!\rightarrow}+w_{\leftarrow}\right)^{2}
			V_{\!\rightarrow}
		}{
			1-\lambda
		}
	}_{F_{\!\leftarrow}
	V_{\!\rightarrow}}
	\,,
\end{equation}
and minimize the misfit,~$\min f$\,, as discussed by \citet{DSSbici1}.
\subsection{Implicit function theorem}
We seek relations between the ratios of quantities on the right-hand side of equation~(\ref{eq:P}).
To do so\,---\,since $f$\,, stated in expression~(\ref{eq:model}), possesses continuous partial derivatives in all its variables at all points, except at ${\lambda=1}$\,, which is excluded by mechanical considerations, and since ${f=0}$\,, as a consequence of equation~(\ref{eq:formula})\,---\,we invoke the implicit function theorem to write
\begin{equation}
	\label{eq:Thm}
	\dfrac{\partial y}{\partial x}
	=
	-\dfrac{
		\dfrac{\partial f}{\partial x}
	}{
		\dfrac{\partial f}{\partial y}
	}
	=:
	-\dfrac{\partial_{x}f}{\partial_{y}f}
	\,,
\end{equation}
where $x$ and $y$ are any two quantities among the arguments of
\begin{equation}
	\label{eq:f}
	f(P,m,g,\theta,a,{\rm C_{rr}},{\rm C_{d}A},\rho,V_{\!\rightarrow},w_{\leftarrow},\lambda)
	\,.
\end{equation}
\subsection{Expressions of partial derivatives}
\label{sub:Partials}
To use formula~(\ref{eq:Thm}), in the context of expression~(\ref{eq:model}), we obtain all partial derivatives of $f$\,, with respect to its arguments.
\begin{equation*}
	\partial_{P}f = 1
\end{equation*}
\begin{equation*}
	\partial_{m}f 
	= 
	-\dfrac{
		\left(
			a 
			+ 
			g\left({\rm C_{rr}}\cos\theta+\sin\theta\right)
		\right)
		V_{\!\rightarrow}
	}{
		1-\lambda
	}
\end{equation*}
\begin{equation*}
	\partial_{\theta}f 
	= 
	-\dfrac{
		m\,g\left(\cos\theta-{\rm C_{rr}}\sin\theta\right)
		V_{\!\rightarrow}
	}{
		1-\lambda
	}
\end{equation*}
\begin{equation*}
	\partial_{a}f = -\dfrac{m\,V_{\!\rightarrow}}{1-\lambda}
\end{equation*}
\begin{equation*}
	\partial_{\rm C_{rr}}f 
	= 
	-\dfrac{m\,g\,\cos\theta\,V_{\!\rightarrow}}{1-\lambda}
\end{equation*}
\begin{equation*}
	\partial_{\rm C_{d}A}f
	= 
	-\dfrac{
		\eta\,\rho
		\left(V_{\!\rightarrow}+w_{\leftarrow}\right)^{2}
		V_{\!\rightarrow}
	}{
		2\,(1-\lambda)
	}
\end{equation*}
\begin{equation*}
	\partial_{\rho}f 
	= 
	-\dfrac{
		\eta\,{\rm C_{d}A}
		\left(V_{\!\rightarrow}+w_{\leftarrow}\right)^{2}
		V_{\!\rightarrow}
	}{
		2\,(1-\lambda)
	}
\end{equation*}
\begin{equation*}
	\partial_{V_{\!\rightarrow}}f
	=
	-\dfrac{
		2\,m\,a
		+
		\eta\,\rho\,{\rm C_{d}A}
		\left(V_{\!\rightarrow}+w_{\leftarrow}\right)
		\left(3\,V_{\!\rightarrow}+w_{\leftarrow}\right)
		+
		2\,m\,g\left({\rm C_{rr}}\cos\theta+\sin\theta\right)
	}{
		2\,(1-\lambda)
	}
\end{equation*}
\begin{equation*}
	\partial_{w_{\leftarrow}}f
	=
	-\dfrac{
		\eta\,\rho\,{\rm C_{d}A}
		\left(V_{\!\rightarrow}+w_{\leftarrow}\right)
		V_{\!\rightarrow}
	}{
		1-\lambda
	}
\end{equation*}
\begin{equation*}
	\partial_{\lambda}f 
	=
	-\dfrac{
		\left(
			2\,m\,a
			+
			\eta\,\rho\,{\rm C_{d}A}
			\left(V_{\!\rightarrow}+w_{\leftarrow}\right)^{2}
			+
			2\,m\,g\left(
				{\rm C_{rr}}\cos\theta+\sin\theta
			\right)
		\right)
		V_{\!\rightarrow}
	}{
		2\,(1-\lambda)^2
	}
\end{equation*}
In accordance with the definition of a partial derivative, all variables in expression~(\ref{eq:f}) are constant, except the one with respect to which the differentiation is performed.
This property is apparent in Appendix~\ref{app:ConstraintPower}, where we examine a relation between differences and derivatives.
\begin{table}
{\small
\begin{center}
{\tabcolsep4pt
{\setlength{\tabcolsep}{15pt}
\renewcommand{\arraystretch}{2.5}
\begin{tabular}{|c||c|c|} \hline
Partial derivative  & Flat course & Inclined course \\
\hline\hline
$\partial_{P}f$ & 1 & 1\\
\hline
$\partial_{m}f$ & 
$-0.246996\,\pm\,0.582875$ & $-2.46314\,\pm\,4.64435$ \\
\hline
$\partial_{\theta}f$ & 
$-11868.6\,\pm\,1108.5$ & $-4636.6\,\pm\,234.084$ \\
\hline
$\partial_a f$ & 
$-1209.85\,\pm\,112.997$ & $-473.421\,\pm\,23.6997$\\
\hline
$\partial_{\rm C_{rr}}f$ & 
$-11868.6\,\pm\,1108.5$ & $-4639.37\,\pm\,233.443$\\
\hline
$\partial_{\rm C_{d}A}f$ & 
$-724.823\,\pm\,203.089$ & $-42.6594\,\pm\,6.38365$\\
\hline
$\partial_{\rho}f$ & 
$-156.937\,\pm\,44.0089$ & $-9.86645\,\pm\,1.47982$\\
\hline
$\partial_hf$ & 
$0.0224033\,\pm\,0.00628245$ & $0.00136659\,\pm\,0.00020498$\\
\hline
$\partial_{V_{\!\rightarrow}}f$ &  
$-56.5462\,\pm\,11.8207$  & $-74.4293\,\pm\,124.542$\\
\hline
$\partial_{w_{\leftarrow}} f$ & 
$-35.9584\,\pm\,6.72944$ & $-5.57109\,\pm\,0.559069$\\
\hline
$\partial_{\lambda}f$ & 
$-224.398\,\pm\,88.3944$ & $-293.684\,\pm\,531.407$\\
\hline
\end{tabular}}}
\end{center}
\caption{\small Values of partial derivatives\,---\,following formul{\ae} in Section~\ref{sub:Partials}\,---\,with the power-meter and GPS measurements collected on a flat and inclined courses}
\label{table:Partials}}
\end{table}
\subsection{Values of partial derivatives}
\subsubsection{Common input values}
To use the partial derivatives stated in Section~\ref{sub:Partials}, we consider measurements collected during two rides.
The flat-course measurements correspond to a nearly flat course.
The inclined-course measurements correspond to an uphill with a nearly constant inclination.
For both the flat and inclined course, we let $m=111$ and $g=9.81$\,.%
\footnote{For consistency with power meters, whose measurements are expressed in watts, which are $\rm{kg\,m^2/s^3}$\,, we use the {\it SI} units for all quantities.
Mass is given in kilograms,~$\rm{kg}$\,, length in metres,~$\rm{m}$\,, and time in seconds,~$\rm{s}$\,; hence, speed is in metres per second, change of speed in metres per second squared, force in newtons,~$\rm{kg\,m/s^2}$\,, and work and energy in joules,~$\rm{kg\,m^{2}/s^{2}}$\,; angle is in radians.}
Other values are stated in Sections~\ref{sec:FlatCourseVals} and \ref{sec:InclCourseVals}.
\subsubsection{Flat-course input values}
\label{sec:FlatCourseVals}
According to \citet{DSSbici1}, the flat-course input values are as follows.
The average measured power, $\overline P=258.8\pm57.3$\,, and speed,~$\overline V_{\!\rightarrow}=10.51\pm0.9816$\,.
The values inferred by modelling are ${\rm C_{d}A}=0.2607\pm0.002982$\,, ${\rm C_{rr}}=0.00231\pm0.005447$ and $\lambda=0.03574\pm0.0004375$\,.
We set $w_{\leftarrow}=0\implies\eta=1$ and $\rho=1.20406\pm0.000764447$\,.
The average slope is $\overline\theta=0.002575\pm0.04027$\,, which indicates a flat course.
The change of speed is $\overline a=0.006922\pm0.1655$\,, which indicates a steady tempo.
Thus, we set $\overline\theta = \overline a = 0$\,.

The corresponding values of partial derivatives, formulated in Section~\ref{sub:Partials}, are listed in the left-hand column of Table~\ref{table:Partials}.
\subsubsection{Inclined-course input values}
\label{sec:InclCourseVals}
\begin{figure}
	\centering
	\includegraphics[scale=0.7]{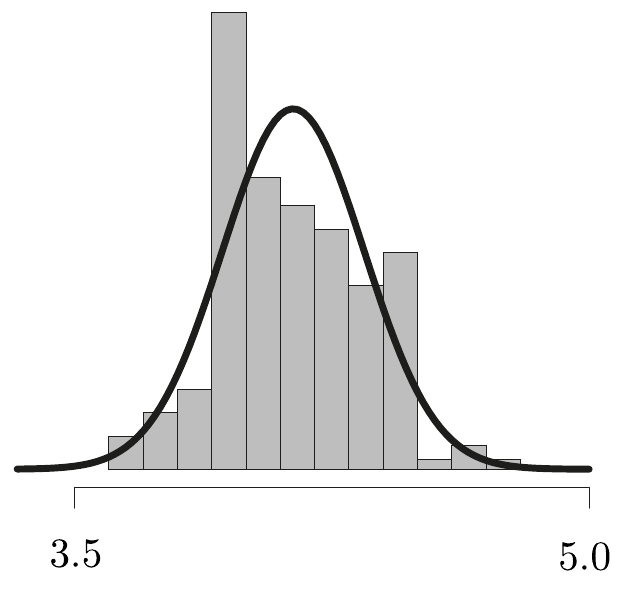}\qquad
	\includegraphics[scale=0.7]{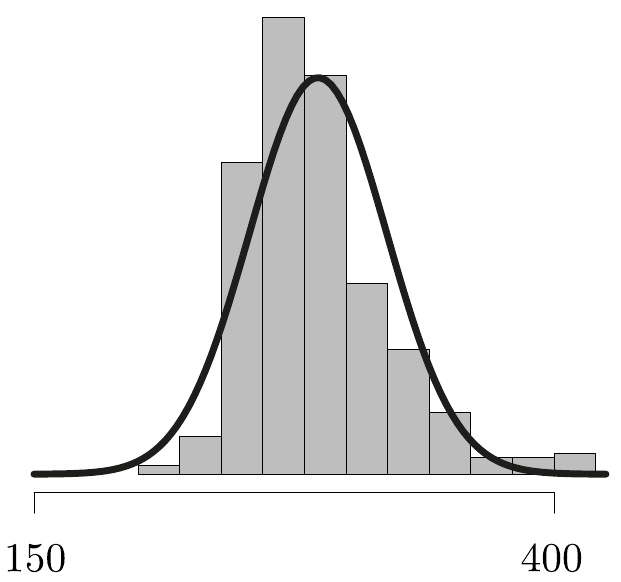}
	\caption{\small
		Left-hand plot: ground speed, from GPS measurements, $\overline V_{\!\rightarrow}=4.138\pm0.2063$\,;
		right-hand plot: power, from power meters, $\overline{P}=286.6\pm33.07$
	}
\label{fig:FigSpeedPowerSteep}
\end{figure}
\begin{figure}
	\centering
	\includegraphics[scale=0.7]{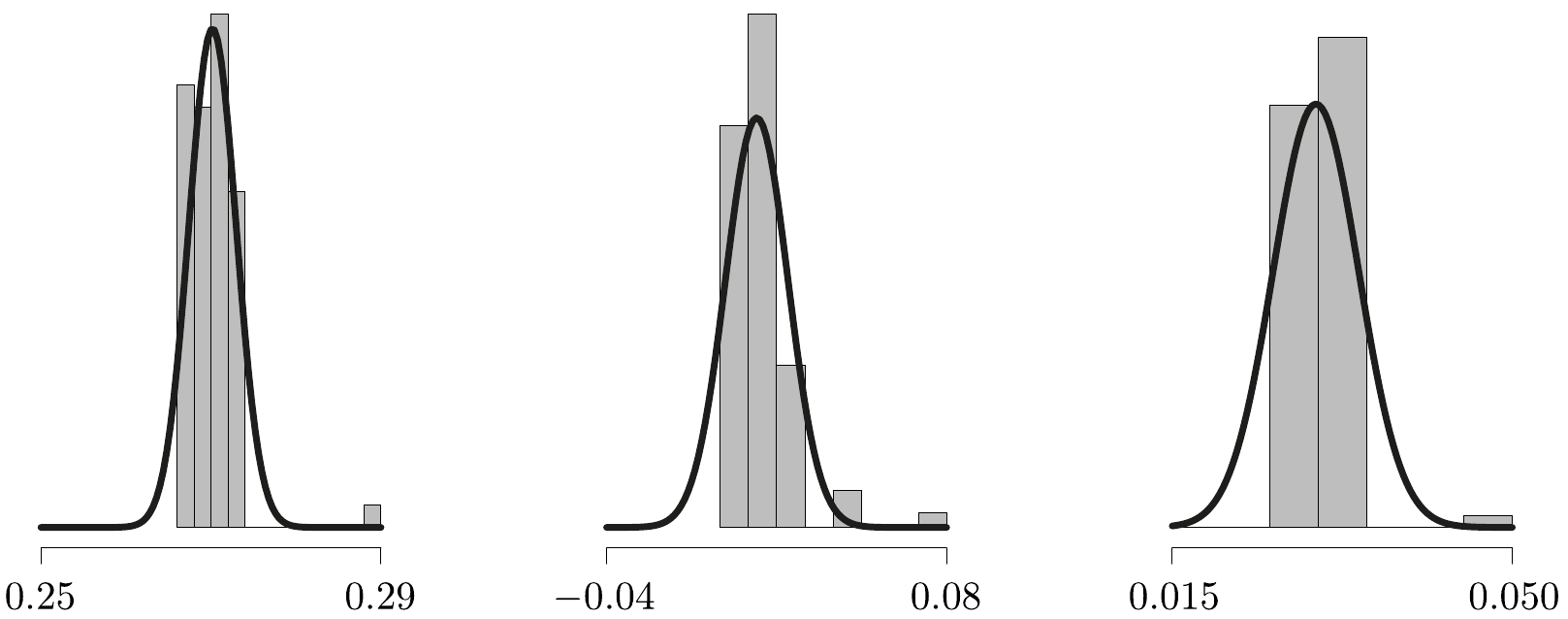}
	\caption{\small
		Optimal values;
		left-hand plot:~${\rm C_{d}A}=0.2702\pm0.002773$\,;
		middle plot:~${\rm C_{rr}}=0.01298\pm0.011$\,;
		right-hand plot:~$\lambda=0.02979\pm0.004396$}
	\label{fig:FigParSteep}
\end{figure}
Following the method outlined by \citet{DSSbici1}, we calculate the inclined-course values.
The data are grouped in eleven speed intervals whose centres range from $3.7$ to $4.7$\,, and which contain three-hundred-and-ninety-two data points.
The mode is $4.2$\,, and is represented by seventy-eight data points.
The distributions of the ground speed and power are illustrated in Figure~\ref{fig:FigSpeedPowerSteep}; their means are $\overline V_{\!\rightarrow}=4.138\pm0.2063$ and $\overline P=286.5783\pm33.11394$\,, respectively.
The distributions of the inferred parameters are illustrated in Figure~\ref{fig:FigParSteep}; their values are ${\rm C_{d}A}=0.2702\pm0.002773$\,, ${\rm C_{rr}}=0.01298\pm0.011$\,, $\lambda=0.02979\pm0.004396$\,.

The average slope of the inclined course is $\overline\theta=0.04592\pm0.1106$\,, which is $2.63^\circ$ and $4.60\%$\,.
This course is known among cyclists of the region as a particularly constant incline.
The change of speed throughout the ride is $\overline a=0.001011\pm0.1015$\,, which indicates a steady pace.
In view of the constantness and steadiness, we set, $\overline\theta=0.04592$ and $\overline a = 0$\,.
Since the change of altitude is negligible\,---\,in the context of air density\,---\,we set $\overline\rho=1.168\pm0.001861$\,, which\,---\,under standard meteorological conditions\,---\,corresponds to the altitude of~$400$\,.
We set $w_{\leftarrow}=0\implies\eta=1$\,.

The corresponding values of partial derivatives, formulated in Section~\ref{sub:Partials}, are listed in the right-hand column of Table~\ref{table:Partials}.
\subsubsection{Theorem requirements}
\begin{figure}
	\centering
	\includegraphics[scale=0.7]{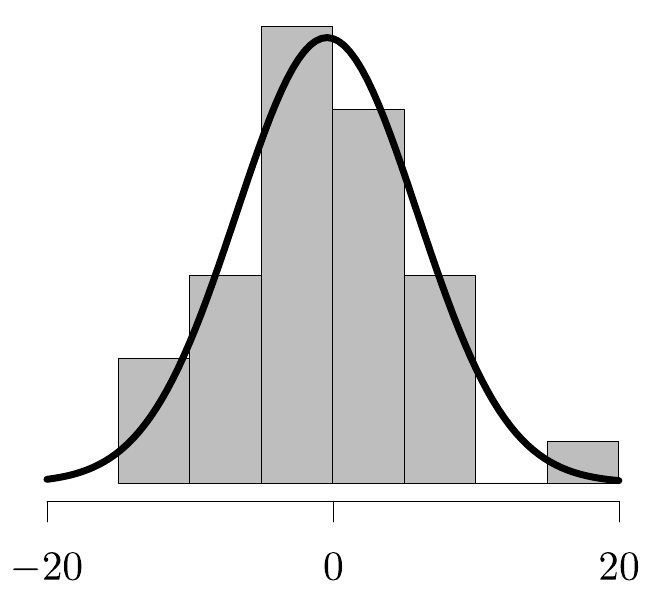}
	\includegraphics[scale=0.7]{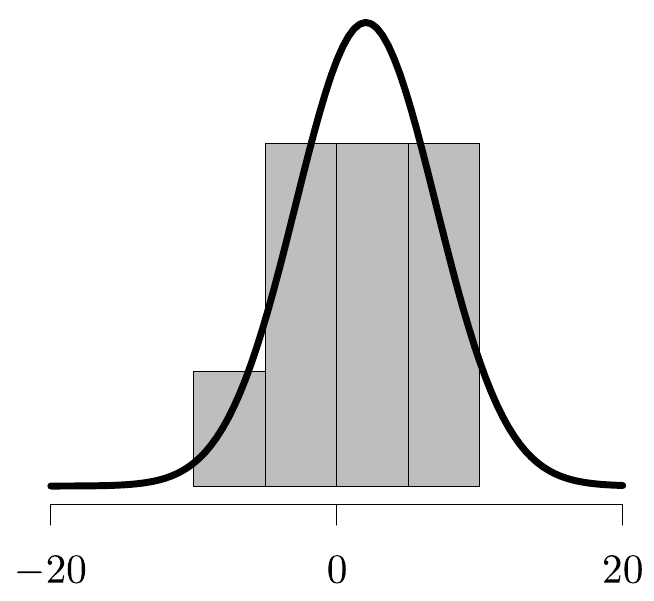}
	\caption{\small Misfit of equation~(\ref{eq:model}): left-hand plot: flat course, $f=0.4137\pm6.321$\,; right-hand plot: inclined course, $f=2.03\pm4.911$}
	\label{fig:Figf}
\end{figure}
As required by the implicit function theorem and as shown in Figure~\ref{fig:Figf}, $f=0$\,, in the neighbourhood of the maxima of the distributions, for both the flat and inclined courses.
Also, as required by the theorem, in formula~(\ref{eq:Thm}), and as shown in Table~\ref{table:Partials}, $\partial_{y}f\neq0$\,, in the neighbourhoods of interest, for either course.

Notably, the similarity of a horizontal spread for both plots of Figure~\ref{fig:Figf} indicates that the goodness of fit of a model is similar for both courses.
The spread is slightly narrower for the inclined course; this might be a result of a lower average speed,~$\overline V_{\!\rightarrow}$\,, which allows for more data points for a given distance and, hence, a higher accuracy of information.
\section{Interpretation}
\subsection{Model considerations}
\label{sec:MC}
The misfit minimization of equation~(\ref{eq:model}), $\min f$\,, treats $\rm C_{d}A$\,, $\rm C_{rr}$ and $\lambda$ as adjustable parameters.
The values in Table~\ref{table:ModelRates} are the changes of $\rm C_{d}A$ due to a change in $\rm C_{rr}$ or $\lambda$\,; in either case, the other quantities are kept constant.
Let us examine the first row.

For the flat course\,---\,in the neighbourhood of $\overline V_{\!\rightarrow}=10.51$ and $\overline P=258.8$\,, wherein ${\rm C_{d}A}=0.2607$ and ${\rm C_{rr}}=0.00231$\,---\,$\partial_{\rm C_{rr}}{\rm C_{d}A}=-16.3745$ and, in accordance with expression~(\ref{eq:Thm}), its reciprocal is $\partial_{\rm C_{d}A}{\rm C_{rr}}=-0.0610705$\,.
We write the corresponding differentials as
\begin{equation*}
	{\rm d(C_{d}A)}
	=
	\dfrac{\partial{\rm C_{d}A}}{\partial{\rm C_{rr}}}\,
	{\rm d(C_{rr})}
	=
	{-16.3745}\,{\rm d(C_{rr})}
\end{equation*}
and
\begin{equation*}
	{\rm d(C_{rr})}
	=
	\dfrac{\partial{\rm C_{rr}}}{\partial{\rm C_{d}A}}\,
	{\rm d(C_{d}A)}
	=
	{-0.0610705}\,{\rm d(C_{d}A)}
	\,;
\end{equation*}
in other words,
an increase of $\rm C_{rr}$ by a unit corresponds to a decrease of  $\rm C_{d}A$ by $16.3745$ units, and an increase of $\rm C_{d}A$ by a unit corresponds to a decrease of  $\rm C_{rr}$ by ${0.0610705}$ of a unit.
Thus,
\begin{equation*}
	\left.
		\dfrac{\rm d(C_{d}A)}{\rm C_{d}A}
	\right|_{{\rm C_{d}A}=0.2607}
	=
	\dfrac{1}{0.2607}
	\,,	
\end{equation*}
which is an increase of about $384\%$\,, corresponds to
\begin{equation*}
	\left.
		\dfrac{\rm d(C_{rr})}{\rm C_{rr}}
	\right|_{{\rm C_{rr}}=0.00231}
	=
	-\dfrac{0.0610705}{0.00231}
	\,,	
\end{equation*}
which is a decrease of about $2644\%$\,.

For the inclined course\,---\,in the neighbourhood of $\overline V_{\!\rightarrow}=4.138$ and $\overline P=286.5783$\,, wherein ${\rm C_{d}A}=0.2702$ and ${\rm C_{rr}}=0.01298$\,---\,$\partial_{\rm C_{rr}}{\rm C_{d}A}=-108.754$ and its reciprocal is $\partial_{\rm C_{d}A}{\rm C_{rr}}=-0.0091951$\,.
Following the same method as for the flat course, we see that an increase of $\rm C_{d}A$ by about $1/0.2702=370\%$ corresponds to a decrease of $\rm C_{rr}$ by about $0.0091951/0.01298=71\%$\,.

Remaining within a linear approximation, an increase of $\rm C_{d}A$ by $1\%$ corresponds to a decrease of $\rm C_{rr}$ by $6.89\%$\,, for the flat course, and a decrease of only $0.19\%$\,, for the inclined course.
This result quantifies that the dependence between $\rm C_{d}A$ and $\rm C_{rr}$\,, within adjustments of the model, is more pronounced for the flat course than for the inclined course, as expected in view of expression~(\ref{eq:P}), whose value---for the inclined course---is dominated by the first summand in the numerator, which includes neither $\rm C_{d}A$ nor $\rm C_{rr}$\,.
This result provides a quantitative justification for the observation that the dependance of the accuracy of the estimate of power on the accuracies of $\rm C_{d}A$ and $\rm C_{rr}$ varies depending on the context; it is more pronounced on flat and fast courses.

Similar evaluations can be performed using the values of derivatives contained in the second row of Table~\ref{table:ModelRates}.
Therein, an increase in $\lambda$ results in a decrease of $\rm C_{d}A$\,, with different rates, for the flat and inclined courses.
\begin{table}
{\small
\begin{center}
{\tabcolsep4pt
{\setlength{\tabcolsep}{15pt}
\renewcommand{\arraystretch}{2.5}
\begin{tabular}{|c||c|c|} \hline
Partial derivative & Flat course & Inclined course \\
\hline\hline
$\partial_{\rm C_{rr}}{\rm C_{d}A}$ & 
$-16.3745\,\pm\,3.05867$ & $-108.754\,\pm\,10.8593$ \\
\hline
$\partial_\lambda{\rm C_{d}A}$ & 
$-0.30959\,\pm\,0.0928393$ & $-6.88438\,\pm\,12.4687$\\
\hline
\end{tabular}}}
\end{center}
\caption{\small Model rates of change following formula~(\ref{eq:Thm}) and values in Table~\ref{table:Partials}}
\label{table:ModelRates}}
\end{table}
\subsection{Physical considerations}
Physical inferences\,---\,based on minimization of expression~(\ref{eq:model})\,---\,are accurate in a neighbourhood of $\overline V_{\!\rightarrow}$ and $\overline P$\,, wherein the set of values for ${\rm C_{d}A}$\,, ${\rm C_{rr}}$ and $\lambda$ is estimated, since, as discussed in Section~\ref{sec:MC}, these values\,---\,in spite of their distinct physical interpretations\,---\,are related among each other by the process of optimization of the model.

In view of expression~(\ref{eq:P}), and as illustrated in Figure~\ref{fig:FigPowerSpeed}, power as a function of ground speed is a cubic.
The inflection point of the curve corresponds to the speed for which there is no air resistance, since the ground speed is equal to the tailwind,~$V_{\!\rightarrow}=-w_{\leftarrow}$\,.
At that point, $P$ is the power to overcome the rolling and drivetrain resistance, only.
To the left of that point, the empirical adequacy of expression~(\ref{eq:P}) is questionable.
However, for the results presented in this article, we consider the cases of~$w\gg -V_{\!\rightarrow}$\,, which are well to the right of the inflection point.
\begin{figure}
	\centering
	\includegraphics[scale=0.7]{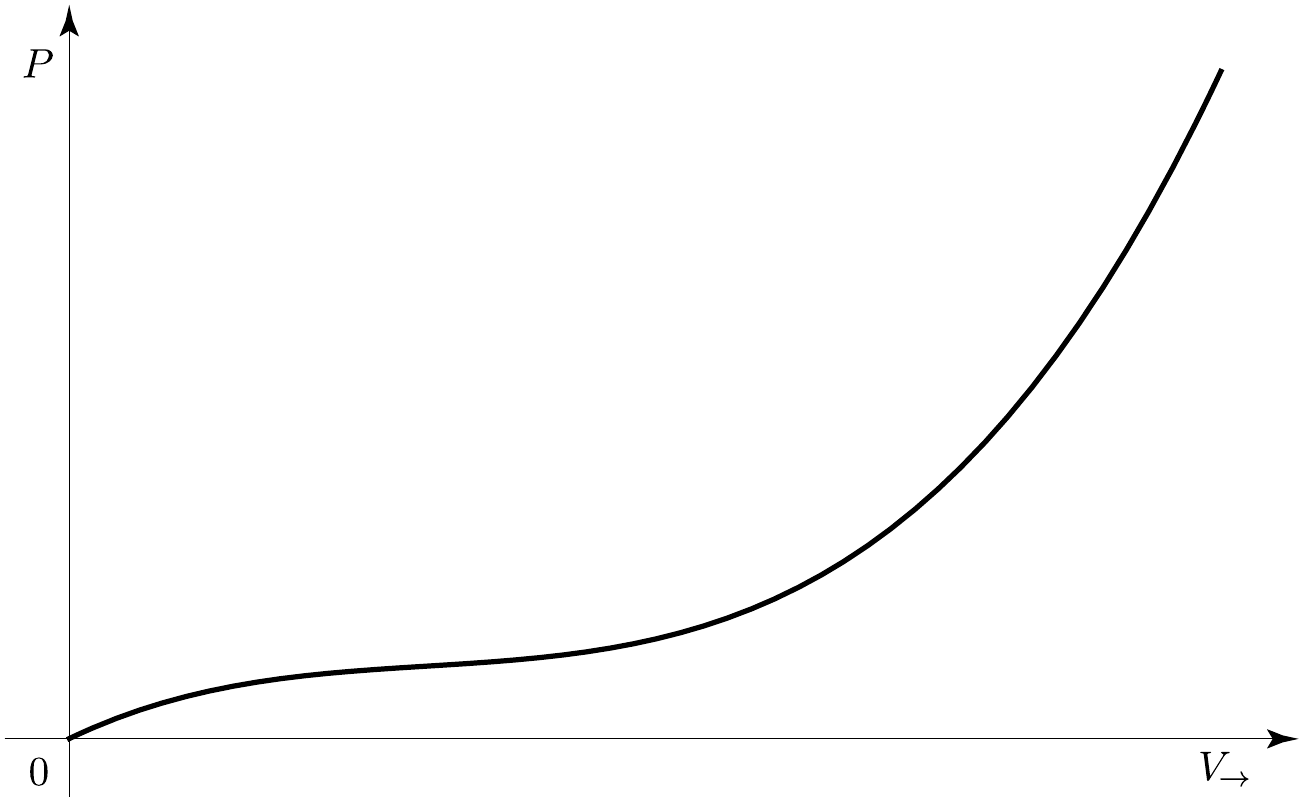}
	\caption{\small Power as a function of speed {\it ceteris paribus}}
	\label{fig:FigPowerSpeed}
\end{figure}

The values in Table~\ref{table:PhysicalRates} are the changes of ground speed due to a change in power, mass, slope and wind; in each case, the other quantities are kept constant.
These values allow us to answer such questions as what increase of speed would result from an increase of power by $1$~watt?
To answer this question, let us examine the first row.

For the flat course\,---\,considering the neighbourhood of $\overline V_{\!\rightarrow}=10.51$ and $\overline P=258.8$\,---\,$\partial_{P}V_{\!\rightarrow}=0.0176846$ and, in accordance with expression~(\ref{eq:Thm}), its reciprocal is $\partial_{V_{\!\rightarrow}}P=56.5462$\,.
We write the corresponding differential as
\begin{equation*}
	{\rm d}P
	=
	\dfrac{\partial P}{\partial V_{\!\rightarrow}}\,
	{\rm d}V_{\!\rightarrow}
	=
	{56.5462}\,{\rm d}V_{\!\rightarrow}
	\,;
\end{equation*}
in other words, an increase of $V_{\!\rightarrow}$ by a unit requires an increase of $P$ by $56.5462$ units.
This means that an increase of speed of $1$~metre per second requires an increase of power of $56.5462$~watts.%
\footnote{$\partial P/\partial V$ can be also found by differentiating expression~(\ref{eq:P}) with respect to $V_{\!\rightarrow}$\,.
However, the implicit function theorem allows us to obtain relations between quantities, without explicitly expressing one in terms of the other.
Also, in accordance with the inverse function theorem, $\partial V/\partial P=1/(\partial P/\partial V)$\,, which is justified by the fact that, for expression~(\ref{eq:P}), $\partial P/\partial V_{\!\rightarrow}\neq 0$\,, in the neighbourhood of interest, as required by the theorem.
However, expression~(\ref{eq:Thm}), which states the implicit function theorem, provides a convenience of examining the relations between the rates of change of any two quantities without invoking the inverse function theorem and requiring an explicit expression for either of them.}

For the inclined course\,---\,considering the neighbourhood of $\overline V_{\!\rightarrow}=4.138$ and $\overline P=286.6$\,---\,$\partial_{P}V_{\!\rightarrow}=0.0134356$ and, in accordance with expression~(\ref{eq:Thm}), its reciprocal is $\partial_{V_{\!\rightarrow}}P=74.4293$\,.
Thus, an increase of speed of about $1$~metre per second requires an increase of power of about $74.4293$~watts.

Hence, for the flat course,
\begin{equation}
	\label{eq:dVV}
	\left.
		\dfrac{{\rm d}V_{\!\rightarrow}}{V_{\!\rightarrow}}
	\right|_{\overline V_{\!\rightarrow}=10.51}
	=
	\dfrac{1}{10.51}
	\,,
\end{equation}
 which is a increase in speed of about $9.5\%$\,, requires
\begin{equation*}
	\left.
		\dfrac{{\rm d}P}{P}
	\right|_{\overline P=258.8}
	=
	\dfrac{56.5462}{258.8}
	\,,
\end{equation*}
which is an increase in power of about $22\%$\,.
For the inclined course, ${{\rm d}V/V}=1/{4.138}$ and ${\rm d}P/P={74.4293}/286.6$\,, which means that a $24\%$ increase in speed requires about $26\%$ increase in power.

Remaining within a linear approximation, an increase of speed by $1\%$ requires an increase of power by about $2.3\%$\,, for the flat course, and an increase of only about $1.1\%$\,, for the inclined course.
This result provides a quantitative justification for a time-trial adage of pushing on the uphills and recovering on the flats, to diminish the overall time.

Since, as illustrated in Figure~\ref{fig:FigPowerSpeed}, the slope of the tangent line changes along the curve, the value of expression~(\ref{eq:Thm}) corresponds to a given neighbourhood of pairs, $\overline V_{\!\rightarrow}$ and $\overline P$\,.
Our interpretation is tantamount to comparing the slopes of two such curves\,---\,one corresponding  to the model of the flat course and the other of the inclined course\,---\,at two distinct locations, $(\overline V_{\!\rightarrow},\overline P)=(10.51,258.8)$ and $(\overline V_{\!\rightarrow},\overline P)=(4.138,286.6)$\,.
Even though the slope of the tangent line changes, it is positive for all values.
This means that the function is monotonically increasing, even though it is a third degree polynomial.
In other words, the relation of power and speed is a bijection, as illustrated in Figure~\ref{fig:FigPowerSpeed} and as discussed in Appendix~\ref{sec:One-to-one}.

Let us examine the second row of Table~\ref{table:PhysicalRates}.
For the flat course, $\partial_{m}V_{\!\rightarrow}=-0.00436803$ and its reciprocal is $\partial_{V_{\!\rightarrow}}m=-228.936$\,.
This means that an increase of speed by $1$~metre per second\,---\,due only to the loss of mass\,---\,requires a decrease of mass of about $229$~kilograms.
For the inclined course, $\partial_{m}V_{\!\rightarrow}=-0.0330937$\,, and its reciprocal is $\partial_{V_{\!\rightarrow}}m=-30.2172$\,; thus, an increase of speed by $1$~metre per second requires a decrease of mass of about $30$~kilograms.

Thus, for the flat course, in accordance with expression~(\ref{eq:dVV}), an increase in speed by about $9.5\%$ requires ${\rm d}m/m=228.936/111$\,, which is a decrease of mass of about $206\%$\,.
For the inclined course, ${{\rm d}V/V}=1/4.138$\,; hence, an increase in speed by about $24\%$ requires ${\rm d}m/m=30.2172/111$\,,  which is a decrease of  mass of about $27\%$\,.
Remaining within a linear approximation, for the flat course, an increase of speed by $1\%$ requires a decrease of mass of about $22\%$\,, and for the inclined course, it requires a decrease of mass of only about $1\%$\,.

This is supportive evidence of an empirical insight into the importance of lightness for climbing; in contrast to flat courses, in the hills, even a small loss of weight results in a noticeable advantage.
Also, this result can be used to quantify the importance of the power-to-weight ratio, which plays an important role in climbing, but a lesser one on a flat.

Similar evaluations can be performed using the values of derivatives contained in the third and fourth rows of Table~\ref{table:PhysicalRates}.
In both cases, the sign is negative; hence, as expected, the increase of steepness or headwind results in a decrease of speed.
These rates of decrease, which are different for the flat and inclined courses, can be quantified in a manner analogous to the one presented in this section.
\begin{table}
{\small
\begin{center}
{\tabcolsep4pt
{\setlength{\tabcolsep}{15pt}
\renewcommand{\arraystretch}{2.5}
\begin{tabular}{|c||c|c|} \hline
Partial derivative & Flat course & Inclined course \\
\hline\hline
$\partial_{P}V_{\!\rightarrow}$ & 
$0.0176846\,\pm\,0.00369688$ & $0.0134356\,\pm\,0.0224816$\\
\hline
$\partial_{m}V_{\!\rightarrow}$ & 
$-0.00436803\,\pm\,0.0098318$ & $-0.0330937\,\pm\,0.00711951$\\
\hline
$\partial_{\theta}V_{\!\rightarrow}$ & 
$-209.893\,\pm\,29.0382$ & $-62.2954\,\pm\,104.671$\\
\hline
$\partial_{w_{\leftarrow}}V_{\!\rightarrow}$ & 
$-0.635912\,\pm\,0.0693924$ & $-0.0748508\,\pm\,0.125421$\\
\hline
\end{tabular}}}
\end{center}
\caption{\small Physical rates of change following formula~(\ref{eq:Thm}) and values in Table~\ref{table:Partials}}
\label{table:PhysicalRates}}
\end{table}
\section{Discussion and conclusions}
The results derived herein are sources of information for optimizing the performance in a time trial under a variety of conditions, such as the strategy of the distribution of effort over the hilly and flat portions or headwind and tailwind sections.
For instance, examining $\partial_{w_{\leftarrow}}V_{\!\rightarrow}$\,, for a flat course, we could quantify another time-trial adage of pushing against the headwind and recovering with the tailwind, to diminish the overall time, under a constraint of cyclist's capacity; such a conclusion is illustrated in Appendix~\ref{sec:Appendix}.
A further insight into this statement is provided by the following example.

Let us consider a five-kilometre section against the headwind,~$w_{\leftarrow}=5$\,, and, following a turnaround, the same five-kilometre section with the tailwind, $w_{\leftarrow}=-5$\,.
If we keep a constant power,~$P=258.8$\,, and use equation~(\ref{eq:P}) to find the corresponding speed, we achieve the total time of $946$~seconds, for ten kilometres, with the upwind speed of $V_{\!\rightarrow}=8.27286$ and the downwind speed of $V_{\!\rightarrow}=14.6269$\,.
If we maintain the same average power, over ten kilometres, but increase the power on the upwind section by $10\%$ and decrease the power on the downwind section by $10\%$\,, we reduce the total time by about $16$~seconds, with the upwind speed of $V_{\!\rightarrow}=8.64701$ and the downwind speed of $V_{\!\rightarrow}=13.7839$\,.
For reliable results\,---\,in view of Figure~\ref{fig:FigPowerSpeed} and the linear approximation within a neighbourhood of the average speed for which the flat-course model is established\,---\,one should not consider excessive increases or decreases of speed or power.
To conclude this example, let us consider the case of keeping a constant speed,~$V_{\!\rightarrow}=11.4256$\,, which is the average for the latter scenario, for ten kilometres.
Such a strategy requires the power for  the upwind section to be $P={531.557}\gg 258.8$\,.
Thus, even though we should push harder against the wind than with the wind, we should not try to keep the same speed for both the upwind and downwind sections.
This conclusion is consistent with the partial-derivative values of Table~\ref{table:PhysicalRates}.

This conclusion is\,---\,only in part\,---\,consistent with a ``Rule of Thumb'' of \citet{Anton2013}.
\begin{quote}
	First, recognize that the equal power outputs recipe, which would have you maintain the same pedal cadence and heart rate in headwind or tailwind, may feel optimal, but it actually isn't.
	In fact, it is only barely faster than suffering the punishing swing in power-output that would be required to maintain equal out-and-back speeds.
	Your overall speed (and your finishing position, of course) will benefit from expending some extra energy when the wind is in your face and conserving some energy when the wind is at your back, but {\it not too much}, because going too far slows you down again as you approach the equal-speeds scenario.	
\end{quote}
A quantification of this, and another, rule of thumb of \citet{Anton2013} is presented in Appendix~\ref{sec:Appendix}, where we question their generality.

Also, results derived in this paper allow for a quantitative evaluation of the aerodynamic efficiency and\,---\,for team time trials\,---\,of the efficiency of drafting.
Under various conditions, there are different relations between the rates of change of quantities in question.
In this paper, as a consequence of the implicit function theorem, relations between the rates of change of all quantities that are included in a model are explicitly stated, and each relation can be evaluated for given conditions.

Furthermore, the derived expressions allow us to interpret the obtained measurements in a quantitative manner, since the values of these expressions entail concrete issues to be addressed for a given bicycle course.
The reliability of information\,---\,which depends on the accuracy of measurements and the empirical adequacy of a model\,---\,is quantified by a misfit and by standard deviations of model parameters.
Also, using partial derivatives listed in Section~\ref{sub:Partials}, we can write the differential of $P$\,, and, hence, estimate its error inherited from the errors of other quantities,
\begin{equation*}
	{\rm d}P
	=
	\partial_{m}P\,{\rm d}m
	+\,\cdots\,+
	\partial_{\lambda}P\,{\rm d}\lambda\,,
\end{equation*}
where, in accordance with equation~(\ref{eq:Thm}) and in view of $\partial_{P}f=1$\,, $\partial_{m}P=-\partial_{m}f,\,\ldots\,,\partial_{\lambda}P=-\partial_{\lambda}f$\,.
\begin{appendix}
\section{Time minimization with Lagrange multipliers}
\label{sec:Appendix}
\setcounter{equation}{0}
\setcounter{figure}{0}
\renewcommand{\theequation}{\Alph{section}.\arabic{equation}}
\renewcommand{\thefigure}{\Alph{section}\arabic{figure}}
\subsection{Preliminary remarks}
Consider a flat course of length~$d$\,, whose one half is covered against the wind, as illustrated in Figure~\ref{fig:FigLagrange}, and the other half with the wind.
\begin{figure}
\centering
\includegraphics[scale=0.75]{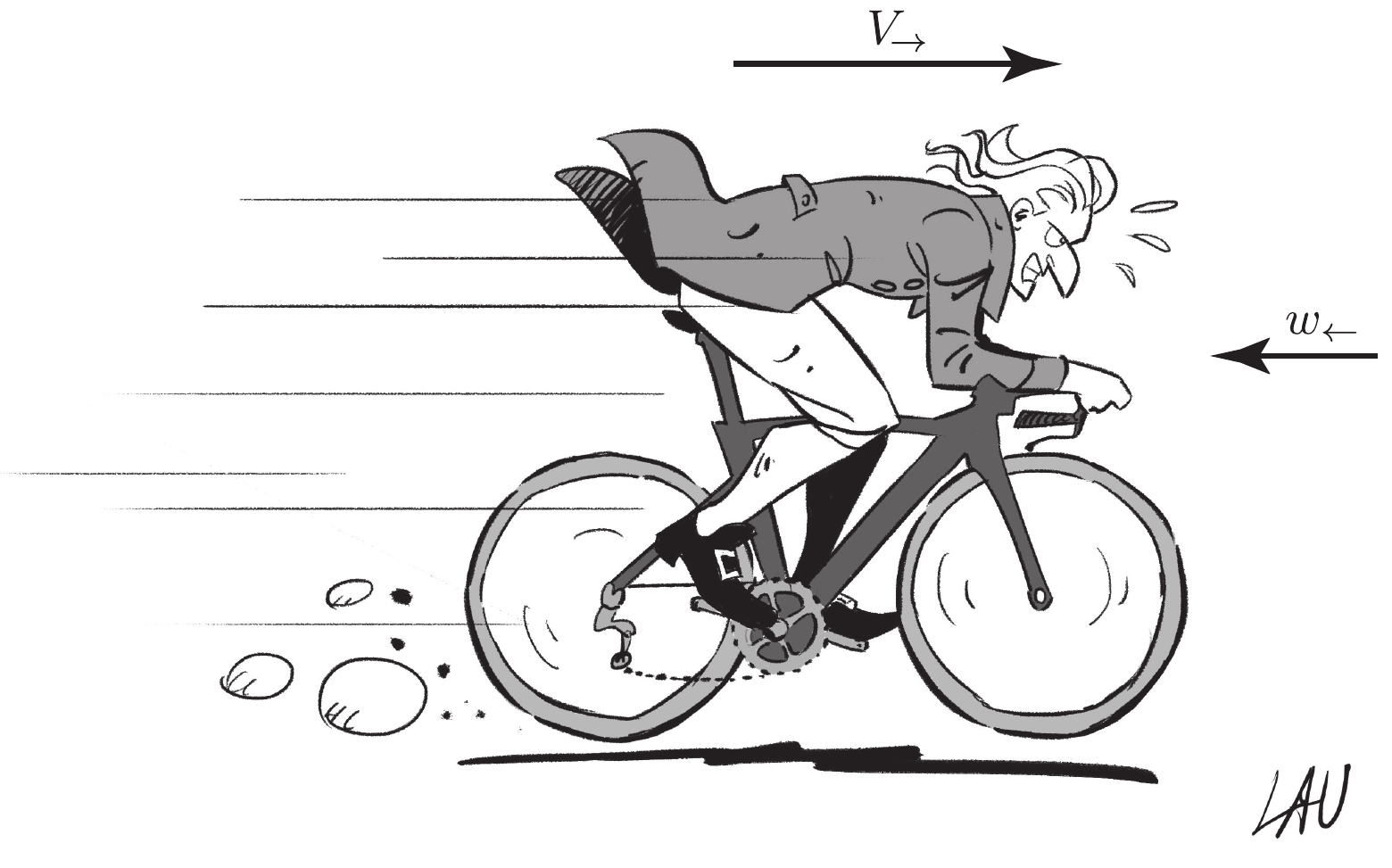}
\caption{\small Joseph-Louis (Giuseppe Luigi) Lagrange examining his optimizations based on the fact that, in general, in accordance with expression~(\ref{eq:P}), headwinds,~$w_\leftarrow>0$\,, increase the air resistance, and tailwinds,~$w_\leftarrow<0$\,, decrease it, provided that $w_{\leftarrow}\nless-V_{\!\rightarrow}$\,.}
\label{fig:FigLagrange}
\end{figure}
To minimize the time,~$t$\,, we need to maximize the average speed,
\begin{equation}
	\label{eq:AveV}
	\overline V_{\!\rightarrow}
	=
	\dfrac{d}{t}
	=
	\dfrac{d}{\dfrac{d}{2\,V_{U}}+\dfrac{d}{2\,V_{D}}}
	=
	\dfrac{2\,V_{U}V_{D}}{V_{U}+V_{D}}
	\,,
\end{equation}
where $V_{U}$ and $V_{D}$ are the speeds on the upwind and downwind sections, respectively.
The maximum of this function occurs for all values along $V_{U}=V_{D}$\,.
To get a pair of values that corresponds to a realistic scenario, we invoke the method of Lagrange multipliers and find the maximum of speed~\eqref{eq:AveV}, subject to constraints.
To do so, we state the problem as a Lagrangian function of two variables with $n$ constraints, 
\begin{equation}
	\label{eq:L_general}
	L(V_{U},V_{D})
	=
	\overline V_{\!\!\rightarrow}(V_{U},V_{D})
	+
	\Lambda_{1}\,\Gamma_{1}(V_{U},V_{D})
	+
	\dots
	+
	\Lambda_{n}\,\Gamma_{n}(V_{U},V_{D})
	\,,	
\end{equation}
where $\Lambda_i$\,, with $i=1,\ldots,n$\,, is a Lagrange multiplier.
The optimization is achieved at the stationary points of function~\eqref{eq:L_general}, which we find by solving the system of equations,
\begin{equation}
	\label{eq:LagPar}
	\dfrac{\partial L}{\partial V_{U}}=0
	\,,\quad
	\dfrac{\partial L}{\partial V_{D}}=0
	\,,\quad
	\dfrac{\partial L}{\partial\Lambda_{1}}= 0
	\,,\quad\ldots\,,\quad
	\dfrac{\partial L}{\partial\Lambda_{n}}=0
	\,,
\end{equation}
whose solution is the pair,~$V_{U},V_{D}$\,, that extremizes expression~(\ref{eq:AveV}) and satisfies the constraints,~$\Gamma_i$\,, where $i=1,\ldots,n$\,, within the physical realm.
\subsection{Constraint of total work}
\label{app:ConstraintWork}
Let us impose a constraint in terms of the amount of total work, $W_{0}=W_{U}+W_{D}$\,, to be done by a cyclist on the upwind and downwind sections, whose proportions of length are stated in expression~(\ref{eq:AveV}),
\begin{equation}
	\label{eq:Gamma_{W}}
	\Gamma_{W}
	=
	\overbrace{
		\underbrace{
			\frac{
				{\rm C_{rr}}\,m\,g
				+
				\tfrac{1}{2}\,{\rm C_{d}A}\,\overline\rho
				\left(V_{U}+w_{\leftarrow}\right)^{2}
			}{
				1-\lambda
			}
		}_{F_{\!\leftarrow}}
		\dfrac{d}{2}
	}^{W_{U}}
	+
	\overbrace{
		\underbrace{
			\frac{
				{\rm C_{rr}}\,m\,g
				+
				\tfrac{1}{2}\,{\rm C_{d}A}\,\overline\rho
				\left(V_{D}-w_{\leftarrow}\right)^{2}
			}{
				1-\lambda
			}
		}_{F_{\!\leftarrow}}
		\dfrac{d}{2}
	}^{W_{D}}-W_{0}
	=
	0
	\,.
\end{equation}
Herein, we assume
\begin{equation*}
	W_{0}
	=
	\underbrace{
		\frac{
			{\rm C_{rr}}\,m\,g
			+
			\tfrac{1}{2}\,{\rm C_{d}A}\,\overline\rho\,{\overline V_{\!\rightarrow}}^{2}
		}{
			1-\lambda
		}
	}_{F_{\!\leftarrow}}
	d
\end{equation*}
to be the total amount of energy available to the cyclist, which corresponds to the work done on the same course, with a maximum effort\,---\,with no wind,~$w_{\leftarrow}=0$\,---\,resulting in a given value of $\overline V_{\!\rightarrow}$\,.
We write function~\eqref{eq:L_general} as
\begin{equation}
	\label{eq:L_{W}}
	L_{W}
	=
	\overline V_{\!\rightarrow}
	+
	\Lambda_{W}\,\Gamma_{W}
	\,.
\end{equation}
Considering $d=10000$\,, model parameters stated in Section~\ref{sec:FlatCourseVals}, namely, $m=111$\,, $g=9.81$\,, $\overline\rho=1.20406$\,, ${\rm C_{d}A}=0.2607$\,, ${\rm C_{rr}}=0.00231$\,, $\lambda=0.03574$\,, and letting $\overline V_{\!\rightarrow}=10.51$\,, we obtain $W_{0}=205878$\,.
To minimize the traveltime with $w_{\leftarrow}=5$\,, we write system~(\ref{eq:LagPar}), in terms of function~(\ref{eq:L_{W}}),
\begin{equation}
	\label{eq:L_{W}_system}
	\begin{dcases}
		\dfrac{\partial L_{W}}{\partial V_{U}}
		=
		\dfrac{2\,{V_{D}}^2}{\left(V_{D} + V_{U}\right)^{2}} 
		+ 
		\Lambda_{W}
		\left(1627.66\,V_{U}+8138.32\right)
		=
		0
		\,,
		\\
		\dfrac{\partial L_{W}}{\partial V_{D}}
		=
		\dfrac{2\,{V_{U}}^2}{\left(V_{D} + V_{U}\right)^{2}}
		+
		\Lambda_{W}
		\left(1627.66\,V_{D}-8138.32\right)
		=
		0
		\,,
		\\
		\dfrac{\partial L_{W}}{\partial\Lambda_{W}}
		= 
		813.832\left({V_{U}}^2 + {V_{D}}^2\right)
		+
		8138.32\left(V_{U}-V_{D}\right)
		-
		139100
		=
		0
		\,.
	\end{dcases}
\end{equation}
Solving system~\eqref{eq:L_{W}_system} numerically, we obtain a single physical solution,
\begin{equation}
	\label{eq:L_{W}_system_sol}
	V_{U} = 8.27945
	\quad{\rm and}\quad
	V_{D} = 11.6766
\end{equation}
which is the pair that both maximizes expression~(\ref{eq:AveV}) and satisfies constraint~(\ref{eq:Gamma_{W}}).

In accordance with expression~(\ref{eq:AveV}), the average speed is $\overline V_{\!\rightarrow} = 9.68886$\,, which is lower than the speed under the assumption of $w_{\leftarrow}=0$\,, namely, $\overline V_{\!\rightarrow}=10.51$\,.
This quantifies an adage that riding with the wind does not compensate for the speed lost by riding against the wind.
The loss is due to the dissipation of energy due to the air, rolling and drivetrain resistances, which are present on both the upwind and downwind sections.
\subsection{Constraint of average power}
\label{app:ConstraintPower}
Let us impose a constraint in terms of the value of average power, $P_0$\,, maintained by a cyclist on the upwind and downwind sections. 
In contrast to work, power is not a cumulative quantity.
Hence, the distance does not appear explicitly in a constraint, and we require constraints for both the upwind and downwind sections,
\begin{align}
	\label{eq:Gamma_{P_{U}}}
	\Gamma_{P_{U}}
	&=
	\overbrace{
		\underbrace{
			\frac{
				{\rm C_{rr}}\,m\,g
				+
				\tfrac{1}{2}\,{\rm C_{d}A}\,\overline\rho
				\left(V_{U}+w_{\leftarrow}\right)^{2}
			}{
				1-\lambda
			}
		}_{F_{\!\leftarrow}}
		V_{U}
	}^{P_{U}}
	-
	P_{0}
	\,,
	\\
	\label{eq:Gamma_{P_{D}}}
	\Gamma_{P_{D}}
	&=
	\overbrace{
		\underbrace{
			\frac{
				{\rm C_{rr}}\,m\,g
				+
				\tfrac{1}{2}\,{\rm C_{d}A}\,\overline\rho
				\left(V_{D}-w_{\leftarrow}\right)^{2}
			}{
				1-\lambda
			}
		}_{F_{\!\leftarrow}}
		V_{D}
	}^{P_{D}}
	-
	P_{0}
	\,.
\end{align}
Herein, we assume
\begin{equation*}
	P_{0}
	=
	\underbrace{
		\frac{
			{\rm C_{rr}}\,m\,g
			+
			\tfrac{1}{2}\,{\rm C_{d}A}\,\overline\rho\,{\overline V_{\rightarrow}}^{2}
		}{
			1-\lambda
		}
	}_{F_{\!\leftarrow}}
	\overline V_{\rightarrow}
\end{equation*}
to be the average power available to the cyclist, which corresponds to the average power achieved on the same course, with a maximum effort\,---\,with no wind,~$w_{\leftarrow}=0$\,---\,resulting in a given value of $\overline V_{\!\rightarrow}$\,.
Function~\eqref{eq:L_general} is
\begin{equation}
	\label{eq:L_{P}}
	L_{P}
	=
	\overline V_{\!\rightarrow}
	+
	\Lambda_{P_{U}}\,\Gamma_{P_{U}}
	+
	\Lambda_{P_{D}}\,\Gamma_{P_{D}}
	\,.
\end{equation}
For $m=111$\,, $g=9.81$\,, $\overline\rho=1.20406$\,, ${\rm C_{d}A}=0.2607$\,, ${\rm C_{rr}}=0.00231$\,, $\lambda=0.03574$\,, $\overline V_{\!\rightarrow}=10.51$\,, we obtain $P_{0} = 216.378$\,.
To minimize the traveltime with $w_{\leftarrow}=5$\,, we write system~(\ref{eq:LagPar}), in terms of function~(\ref{eq:L_{P}}),
\begin{equation*}
	\begin{dcases}
		\dfrac{\partial L_{P}}{\partial V_{U}}
		=
		\dfrac{2\,{V_{D}}^2}{\left(V_{D} + V_{U}\right)^{2}} 
		+ 
		\Lambda_{P_{U}}
		\left(0.488299\,{V_{U}}^2+3.25533\,V_{U}+6.67778\right)
		=
		0
		\,,
		\\
		\dfrac{\partial L_{P}}{\partial V_{D}}
		=
		\dfrac{2\,{V_{U}}^2}{\left(V_{D} + V_{U}\right)^{2}}
		+
		\Lambda_{P_{D}}
		\left(0.488299\,{V_{D}}^2-3.25533\,V_{D}+6.67778\right)
		=
		0
		\,,
		\\
		\dfrac{\partial L_{P}}{\partial\Lambda_{P_{U}}}
		=
		0.162766\,{V_{U}}^{3} 
		+ 
		1.62766\,{V_{U}}^{2} 
		+ 
		6.67778\,V_{U}
		-
		216.378
		=
		0
		\,,
		\\
		\dfrac{\partial L_{P}}{\partial\Lambda_{P_{D}}}
		= 
		0.162766\,{V_{U}}^{3} 
		- 
		1.62766\,{V_{U}}^{2} 
		+ 
		6.67778\,V_{U} 
		-
		216.378
		=
		0
		\,.
	\end{dcases}
\end{equation*}
The single physical solution is
\begin{equation}
	\label{eq:L_{P}_system_sol}
	V_{U} = 7.60272
	\quad{\rm and}\quad
	V_{D} = 13.9163
	\,,
\end{equation}
which both maximizes expression~\eqref{eq:AveV} and satisfies constraints~\eqref{eq:Gamma_{P_{U}}} and~\eqref{eq:Gamma_{P_{D}}}.
The corresponding average is $\overline V_{\!\!\rightarrow} = 9.83332$\,, which confirms an adage that riding with the wind does not compensate for the speed lost by riding against the wind.
\subsection{Relation between differences and derivatives}
To conclude this appendix, let us comment on difference~$\Delta V_{\!\rightarrow}/\Delta w_{\leftarrow}$\,, discussed herein, in the context of $\partial_{w_{\leftarrow}}V_{\!\rightarrow}$\,, whose value is presented in Table~\ref{table:PhysicalRates}.
Partial derivatives correspond to a tangent to a curve at a point, and the differences to a secant over a segment of the curve.
Also, partial derivatives are obtained under the assumption that all other quantities are constant.

The latter requirement is satisfied in Appendix~\ref{app:ConstraintPower}, where
\begin{equation*}
	\frac{\Delta V_{\rightarrow}}{\Delta w_{\leftarrow}}
	=
	\frac{V_{U}-V_{D}}{w_{\leftarrow}-(-w_{\leftarrow})}
	=
	\frac{7.60272-13.9163}{10}
	=
	-0.631356\,,
\end{equation*}
which agrees with $\partial_{w_{\leftarrow}}V_{\rightarrow}$\,, in Table~\ref{table:PhysicalRates}, to two decimal points.
For $w_{\leftarrow}=0.05$\,, we obtain $V_{U} = 10.4782$ and $V_{D} = 10.5418$\,; hence, $\Delta V_{\rightarrow}/\Delta w_{\leftarrow}=-0.635911$\,, which agrees with $\partial_{w_{\leftarrow}}V_{\rightarrow}$ to six decimal points.
In general,
\begin{equation*}
	\lim_{\Delta w_{\leftarrow}\rightarrow0}
	\frac{\Delta V_{\rightarrow}}{\Delta w_{\leftarrow}}
	=
	\frac{\partial V_{\rightarrow}}{\partial w_{\leftarrow}}
	\,,
\end{equation*}
as expected, in view of a secant approaching a tangent.

The requirement of constant quantities is not satisfied in Appendix~\ref{app:ConstraintWork}, since $P$ is allowed to vary to maintain the imposed value of~$W_{0}$\,.
In Appendix~\ref{app:ConstraintPower}, $W$ varies to maintain the imposed value of~$P_0$\,, but $W$ is not a variable in function~(\ref{eq:f}), used in partial derivatives.

As shown in this appendix, properties of partial derivatives need to be considered in examining time-trial strategies.
In contrast to common optimization methods, partial derivatives correspond to a change of a single variable, only.
\subsection{Closing remarks}
Let us examine the constraints discussed in this appendix in terms of required powers.
For the work constraint, the average speed is~$\overline V_{\!\rightarrow}= 9.68886$\,.
Following expression~(\ref{eq:P}), the required powers are $P_{U}=365.537$ and $P_{D}=59.9459$\,, for the upwind and downwind sections, respectively.
For windless conditions, we have $\overline{P} =173.316$\,.
Thus, $P_{U}$ is significantly greater than~$\overline P$\,.

For the power constraint, with $P_{0}=216.378$\,, the average speed is~$\overline V_{\!\rightarrow}= 9.83332$\,.
Since the average speed is greater than for the work-constraint optimization and the average power does not exceed the value obtained in windless conditions, this appears to be the preferable strategy.
Also, power is provided as an instantaneous quantity by the power meters, which allows the cyclist to follow a given strategy, whose further refinements are to be considered in future studies.

To close, let us consider a rule of thumb of~\citet{Anton2013}.
\begin{quote}
	Choose a target-speed~$v_0$\,.
	$\,[\ldots]\,$
	Endeavor to ride at $v\cong v_0+w/4$ when the wind is at your back and at
	$v\cong v_0-w/2$ when the wind is at your face.	
\end{quote}
If we choose $\overline V_{\!\rightarrow}=10.51=:v_0$ to be a target speed, with $w=5$\,, speeds~\eqref{eq:L_{P}_system_sol}, which result from the power constraint, are less congruent with this rule than speeds~\eqref{eq:L_{W}_system_sol}, which result from the work constraint, yet---according to the present analysis---speeds~\eqref{eq:L_{P}_system_sol} appear to be preferable.
This is an indication of further subtleties that need to be considered in developing a time-trial strategy.

\section{Relation between power and speed}
\label{sec:One-to-one}
\setcounter{equation}{0}
\setcounter{figure}{0}
\renewcommand{\theequation}{\Alph{section}.\arabic{equation}}
\renewcommand{\thefigure}{\Alph{section}\arabic{figure}}
\subsection{Proposition}
\label{sub:Proposition}
\begin{proposition}
\label{prop:one-to-one}
According to model~(\ref{eq:P}), with $a=0$ and $\eta=1$\,, the relation between the generated power,~$P$\,, and the bicycle speed,~$V_{\!\rightarrow}$\,, is one-to-one.	
\end{proposition}
\begin{proof}
It suffices to show that $\partial P/\partial V_{\!\rightarrow}>0$\,, for $V_{\!\rightarrow}\in(0,\infty)$\,.
Since
\begin{equation}
\label{eq:Collorary}
\dfrac{\partial P}{\partial V_{\!\rightarrow}}
=\frac{({\rm C_{rr}}\cos\theta+\sin\theta)\,g\,m+\frac{3}{2}{\rm C_dA}\,\rho\,V_{\!\rightarrow}^{\,2}}{1-\lambda}\,,
\end{equation}
where all quantities are positive and $\lambda\ll 1$\,, it follows that $\partial P/\partial V_{\!\rightarrow}>0$ and, hence, the relation between power and speed is one-to-one.
\end{proof}
This bijection means that the power generated by a cyclist and the bicycle speed are related by an invertible function, which is consistent with unique physical solutions obtained in Appendices~\ref{app:ConstraintWork} and \ref{app:ConstraintPower}.
Also, for expression~(\ref{eq:Collorary}), $\partial^2P/\partial V_{\!\rightarrow}^{\,2}>0$\,; hence\,---\,in contrast to Figure~\ref{fig:FigPowerSpeed}, where $\eta=\pm 1$\,---\,the curve corresponding to $P$ as a function of $V_{\!\rightarrow}$\,, with $\eta=1$\,, is concave up for $V_{\!\rightarrow}\in(0,\infty)$\,.
Let us state two corollaries of Proposition~\ref{prop:one-to-one}.
\medskip
\begin{corollary}
Ceteris paribus, an increase of speed requires an increase of power, and an increase of power results in an increase of speed.
\end{corollary}
As illustrated in Figure~\ref{fig:FigPowerSpeed}, this remains true even for $\eta=-1$\,.
Also,
\medskip
\begin{corollary}
\label{corr:2}
Ceteris paribus, a constant power results in a constant speed, an increase of power results in an instantaneous increase of speed, a decrease of power results in an instantaneous decrease of speed.
\end{corollary}
This corollary, which is illustrated by \citet[Figure~17]{BSSS}, for a velodrome, remains valid for a cyclist riding uphill and experiencing resistive forces.
However, it does not hold for a cyclist going down a sufficiently steep hill,~$\theta\ll0$\,.

Proposition~\ref{prop:one-to-one} and its corollaries are valid in the context of the instantaneous power generated by a cyclist.
As discussed in Appendix~\ref{sub:Qualifier}, this need not be tantamount to the power measured at that instant.
\subsection{Fixed-wheel drivetrain}
\label{sub:Qualifier}
\begin{figure}
	\centering
	\includegraphics[scale=0.5]{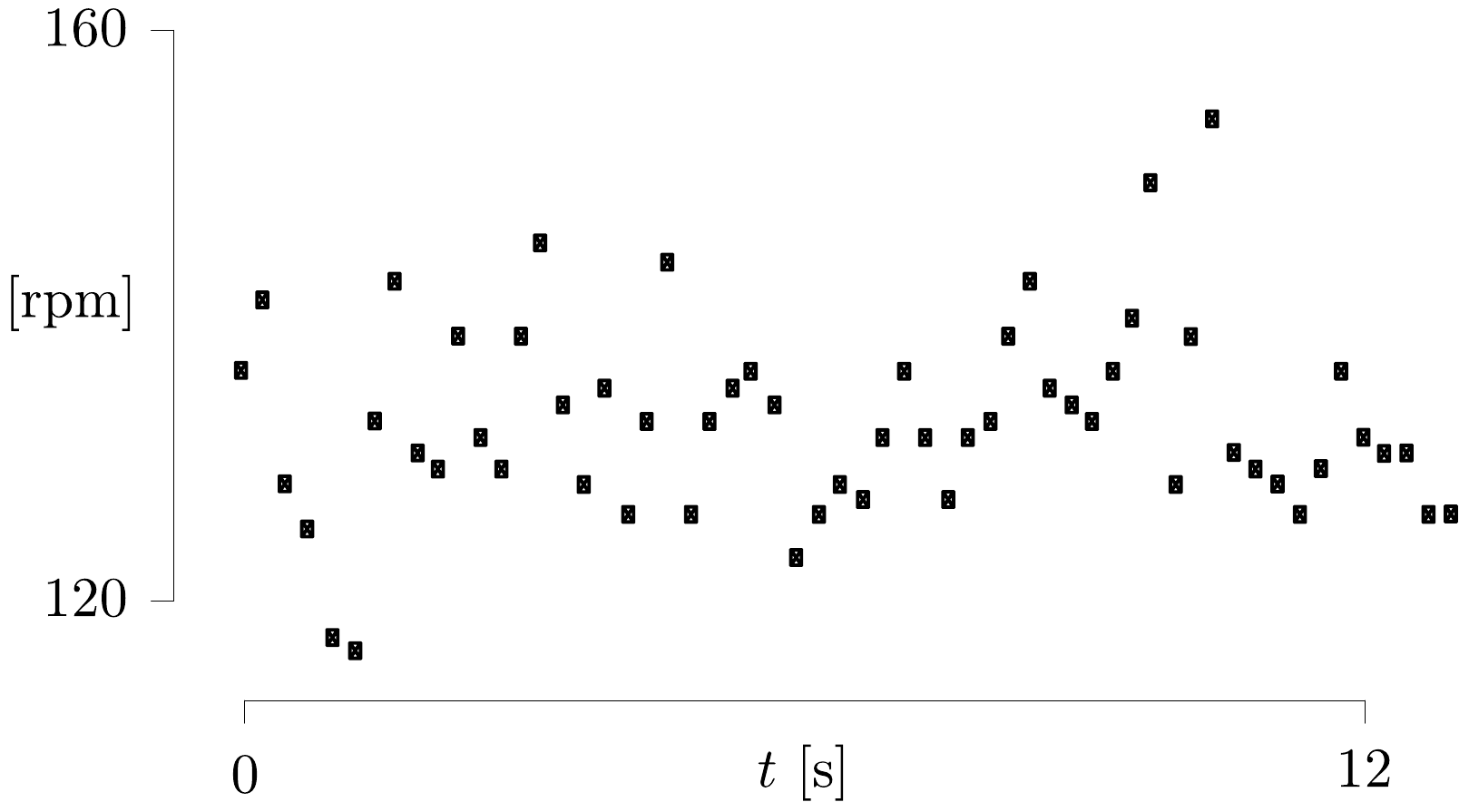}
	\caption{\small Measured cadence}
	\label{fig:FigCadence}
\end{figure}
Let us illustrate the difference between the power generated instantaneously by a cyclist and the measured power, for a fixed-wheel drivetrain.
Power,~$P$\,, is obtained from the measurements of $f_{\circlearrowright}$\,, which is the force applied to pedals, and $v_{\circlearrowright}$\,, which is the circumferential speed%
\footnote{Given a crank length, $v_{\circlearrowright}$ is obtained from measurements of cadence.}
of rotating pedals \citep[e.g.,][expression~(1)]{DSSbici1},
\begin{equation}
	\label{eq:formula2}
	P=f_{\circlearrowright}\,v_{\circlearrowright}\,.
\end{equation}
The resulting value of $P$ corresponds to the power generated instantaneously by a  cyclist if $v_{\circlearrowright}$ is an instantaneous consequence of $f_{\circlearrowright}$\,, only, as is the case of a free-wheel drivetrain; if no force is applied to pedals, the pedals do not rotate, regardless of the bicycle speed.
For a fixed-wheel drivetrain, for which there is a one-to-one relation between $v_{\circlearrowright}$ and the wheel speed, the momentum of the bicycle-cyclist system\,---\,which results in a pedal rotation, even without any force applied by a cyclist\,---\,might affect the value of~$v_{\circlearrowright}$\,.
\begin{figure}
	\centering
	\includegraphics[scale=0.5]{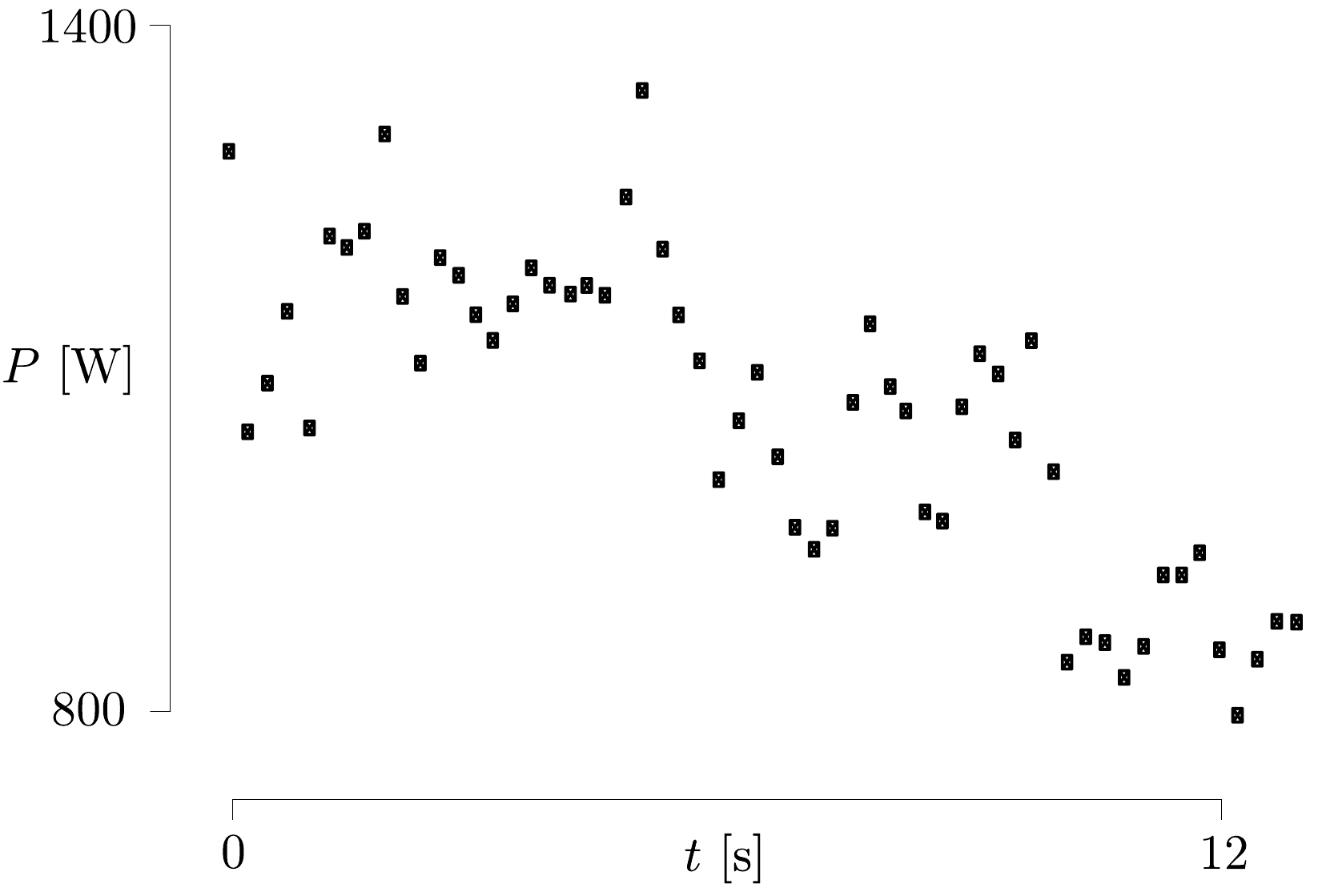}
	\caption{\small Measured power,~$P$}
	\label{fig:FigPower}
\end{figure}

To illustrate this issue, let us examine Figures~\ref{fig:FigCadence}, \ref{fig:FigPower} and \ref{fig:FigForce}, which correspond to the fixed-wheel measurements of, respectively, cadence, power and force from the second lap of a $1000$\,-metre time trial~(Mehdi Kordi, {\it pers.~comm.}, 2020).
Taking into account inaccuracies of measurements, the first figure exhibits a steadiness of cadence,%
\footnote{Strictly speaking, unless a sufficient force is applied to the pedals, there is necessarily a decrease of cadence.
 Only hypothetically\,---\,with no internal or external resistance on the flats\,---\,the cadence would remain constant.}
 which entails the steadiness of~$v_{\circlearrowright}$ and the steadiness of the wheel speed.
The second figure exhibits a decrease of power, which\,---\,in view of the steadiness of~$v_{\circlearrowright}$ and in accordance with expression~(\ref{eq:formula2})\,---\,is a consequence of the decrease of~$f_{\circlearrowright}$\,, shown in the third figure.
The behaviours observed in these figures are confirmed by the Augmented Dickey-Fuller test \citep{DickeyFuller}.
However, $F_{\!\leftarrow}$\,, in expression~(\ref{eq:formula}), need not decrease, since\,---\,in view of the steadiness of the wheel speed\,---\,$V_{\!\rightarrow}$ is approximately steady,%
\footnote{On a velodrome, $V_{\!\rightarrow}$ is equivalent to the wheel speed along the straights but not along the curves, due to the leaning of the bicycle-cyclist system.
Along the curves, and in general, $V_{\!\rightarrow}$ corresponds to the centre-of-mass speed~\citep{BSSS}.}
and so can be the power of the bicycle-cyclist system, for which we do not have direct measurements.
\begin{figure}
	\centering
	\includegraphics[scale=0.5]{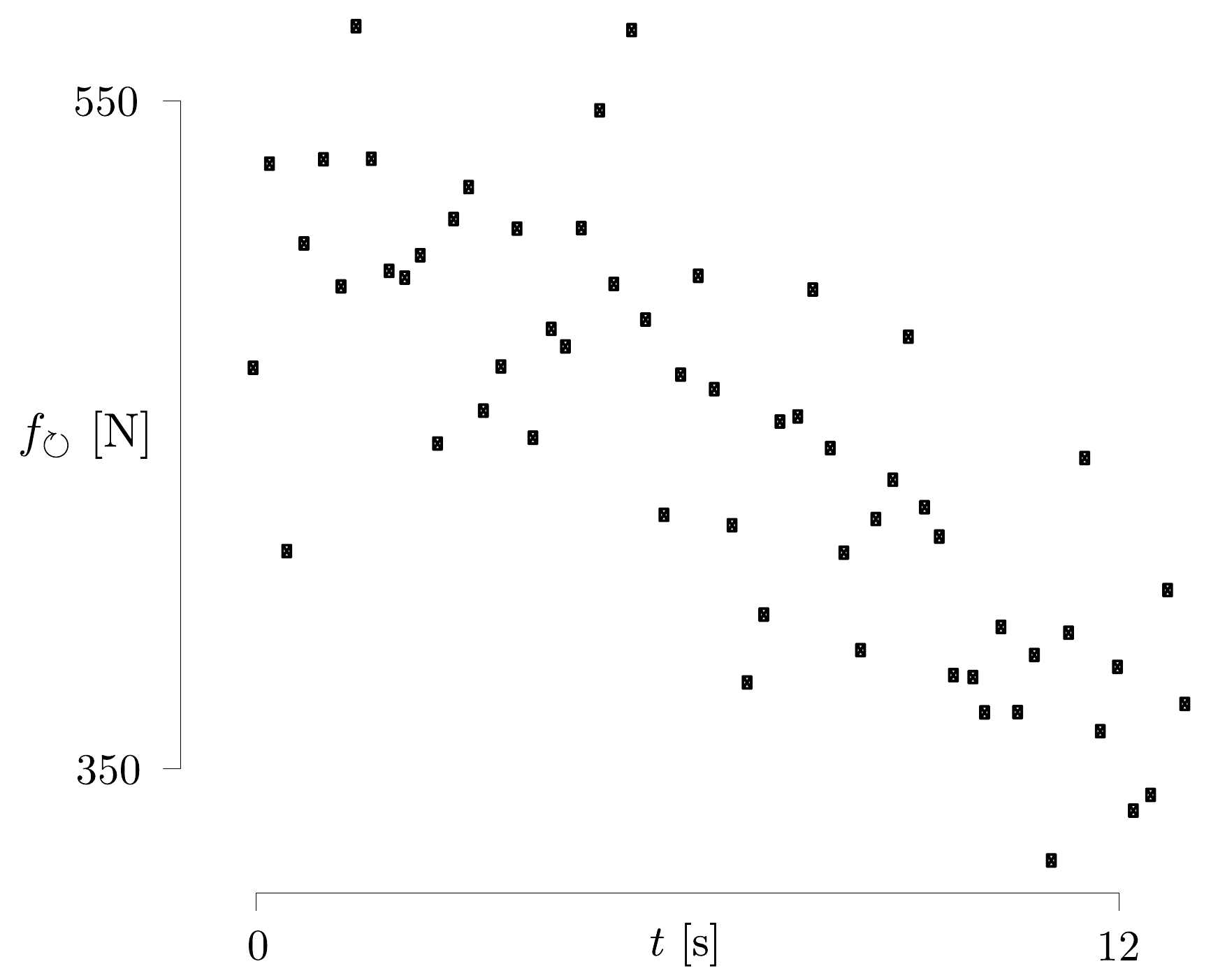}
	\caption{\small Measured force,~$f_{\circlearrowright}$}
	\label{fig:FigForce}
\end{figure}

In general, the measurements of power connected to the drivetrain, and the power of the system itself are distinct from one another.
For instance, for a free-wheel drivetrain, on a downhill on which a cyclist does not pedal, $f_{\circlearrowright}\,v_{\circlearrowright}=0\neq F_{\!\leftarrow}\,V_{\!\rightarrow}>0$\,.%
\footnote{On a downhill, due to gravitation, the bicycle-cyclist system can accelerate even if the cyclist does not apply any force to the pedals.
On a flat, {\it ceteris paribus}, the system, albeit gradually, must slow down.}
Similarly, for a fixed-wheel drivetrain, the instantaneous measurements of power need not represent the power generated by a cyclist at these instances, since the rotation of the pedals might be partially due to the force exerted by the cyclist, at a given moment, and partially due to the momentum of the already moving bicycle-cyclist system.
The issue remains even if the sensors are not incorporated within the pedals but in the cranks, a bottom bracket or a rear hub.
In each case, the sensors are connected to the drivetrain.

Thus, for a fixed-wheel drivetrain, the power-meter measurements represent an instantaneous power generated by a cyclist if the power and cadence are in a dynamic equilibrium.
Such an equilibrium is reachable\,---\,following an initial acceleration\,---\,during a steady effort, as is the case of a $4000$\,-metre individual pursuit \citep[e.g.,][Section~5]{BSSS} or the Hour Record, in contrast to accelerations followed by decelerations.
In the context of expression~(\ref{eq:formula2}), a dynamic equilibrium means that changes of $v_{\circlearrowright}$ are immediate responses solely to changes in~$f_{\circlearrowright}$\,.  
\section*{Acknowledgements}
We wish to acknowledge Len Bos and Rapha\"el Slawinski, for fruitful discussions, Mehdi Kordi for providing the measurements, used in Appendix~\ref{sub:Qualifier}, David Dalton, for his scientific editing and proofreading, Elena Patarini, for her graphic support, and Roberto Lauciello, for his artistic contribution. 
Furthermore, we wish to acknowledge Favero Electronics for inspiring this study by their technological advances and for supporting this work by providing us with their latest model of Assioma Duo power meters.
\section*{Conflict of Interest}
The authors declare that they have no conflict of interest.
\bibliographystyle{spbasic}
\bibliography{DSSbici2_arXiv.bib}
\end{appendix}
\end{document}